\documentclass[11pt,english,preprint]{revtex4-1}
\usepackage{ae,aecompl}
\usepackage[T1]{fontenc}
\usepackage[latin9]{inputenc}
\usepackage{color}
\usepackage{babel}
\usepackage{mathtools}
\usepackage{bm}
\usepackage{amsmath}
\usepackage{amsthm}
\usepackage{amssymb}
\usepackage{mathdots}
\usepackage{stmaryrd}
\usepackage{stackrel}
\usepackage{undertilde}
\PassOptionsToPackage{version=3}{mhchem}
\usepackage{mhchem}
\usepackage{esint}
\usepackage[unicode=true,pdfusetitle,
 bookmarks=true,bookmarksnumbered=false,bookmarksopen=false,
 breaklinks=true,pdfborder={0 0 1},backref=false,colorlinks=true]
 {hyperref}
\hypersetup{
 urlcolor=red, citecolor=blue, linkcolor=black}
\usepackage{breakurl}

\makeatletter

\providecommand{\tabularnewline}{\\}

\theoremstyle{definition}
\newtheorem*{defn*}{\protect\definitionname}
\theoremstyle{plain}
\newtheorem{prop}{\protect\propositionname}
\theoremstyle{plain}
\newtheorem{fact}{\protect\factname}
\theoremstyle{remark}
\newtheorem*{rem*}{\protect\remarkname}
\theoremstyle{plain}
\newtheorem{lem}{\protect\lemmaname}
\theoremstyle{plain}
\newtheorem{cor}{\protect\corollaryname}
\theoremstyle{plain}
\newtheorem{thm}{\protect\theoremname}
\theoremstyle{remark}
\newtheorem*{acknowledgement*}{\protect\acknowledgementname}

\usepackage{babel}

\makeatother

\providecommand{\acknowledgementname}{Acknowledgement}
\providecommand{\corollaryname}{Corollary}
\providecommand{\definitionname}{Definition}
\providecommand{\factname}{Fact}
\providecommand{\lemmaname}{Lemma}
\providecommand{\propositionname}{Proposition}
\providecommand{\remarkname}{Remark}
\providecommand{\theoremname}{Theorem}

\begin{document}

\title{Optimal two-qubit states for quantum teleportation vis-à-vis state
properties}

\author{Arkaprabha Ghosal}
\email{a.ghosal1993@gmail.com}

\affiliation{Centre for Astroparticle Physics and Space Science, Bose Institute,
Block EN, Sector V, Salt Lake, Kolkata 700 091, India}

\author{Debarshi Das}
\email{dasdebarshi90@gmail.com}

\affiliation{Centre for Astroparticle Physics and Space Science, Bose Institute,
Block EN, Sector V, Salt Lake, Kolkata 700 091, India}

\author{Saptarshi Roy}
\email{saptarshiroy@hri.res.in}

\affiliation{Harish-Chandra Research Institute, HBNI, Chhatnag Road, Jhunsi, Allahabad
211 019, India}

\author{Somshubhro Bandyopadhyay}
\email{som@jcbose.ac.in}

\affiliation{Department of Physics and Centre for Astroparticle Physics and Space
Science, Bose Institute, Block EN, Sector V, Salt Lake, Kolkata 700
091, India}
\begin{abstract}
Quantum teleportation with a two-qubit state can be suitably characterized
in terms of maximal fidelity and fidelity deviation -- the former
is the maximal value of the average fidelity achievable within the
standard protocol and local unitary operations and the latter is the
standard deviation of fidelity over all input states. In this paper,
we consider the problem of characterizing two-qubit states that are
optimal for quantum teleportation for a given value of some state
property. The optimal states are defined as those states that, for
a given value of the state property under consideration, achieve the
largest maximal fidelity and also exhibit zero fidelity deviation.
We provide a complete characterization of optimal states for a given
linear entropy, maximum mean value of the Bell-CHSH observable, and
concurrence, respectively. We find that for a given linear entropy
or Bell-CHSH violation, the largest maximal fidelity states are optimal,
but for a given concurrence, the optimal states form a strict subset
of the largest maximal fidelity states. 
\end{abstract}
\maketitle

\section{Introduction}

Quantum teleportation \citep{teleportation-93} is a fundamental protocol
to transmit quantum information using shared entanglement and classical
communication. Perfect teleportation requires maximally entangled
states which can only be established via noiseless quantum channels.
In practice, however, quantum channels are noisy; thus available states
are typically mixed entangled and teleportation will not be perfect.
The standard figure of merit for quantum teleportation is the average
fidelity \citep{RMP-Horodecki-96,Gisin-1996,MPR-Horodecki-1999,VV-2003},
which is a measure of the expected closeness between the input and
the output states. The average fidelity is a pretty good indicator
of how useful a given entangled state is for quantum teleportation,
but gives no information on fidelity fluctuations. Such fluctuations
can be appropriately quantified by fidelity deviation, which is defined
as the standard deviation of fidelity over all input states \citep{Bang-et-al-2012,Bang-et-al-2018}.
The notions of average fidelity and fidelity deviation are completely
general and apply to quantum teleportation in any finite dimension,
although in this paper we restrict ourselves to quantum teleportation
with two-qubit states. 

For a two-qubit state, the maximal fidelity \citep{RMP-Horodecki-96,Badziag-2000}
is the maximal value of the average fidelity achievable over all strategies
within the standard protocol supplemented by local unitary operations
and is given by a simple formula, first derived in \citep{RMP-Horodecki-96}.
A protocol, which achieves this maximal value is said to be optimal.
Fidelity deviation, on the other hand, is something one would like
to minimize but not at the expense of maximal fidelity, as pointed
out in \citep{Ghosal-deviation-2019}. The authors argued that a protocol
which achieves the minimal fidelity deviation is not guaranteed to
achieve the maximal fidelity; hence, one should compute fidelity deviation
only for optimal protocols; in particular, they obtained an exact
formula for the fidelity deviation in optimal quantum teleportation
\citep{Badziag-2000} with a two-qubit state. They further showed
that fidelity deviation is nonzero for generic two-qubit states but
there exist states, other than maximally entangled and Werner \citep{Bang-et-al-2018},
with vanishing fidelity deviation. 

While maximal fidelity indicates how well, on average, an input state
is teleported, fidelity deviation is a measure of dispersion. But
considered together, the ordered pair is expected to serve as a better
performance measure, one which is more effective than maximal fidelity
alone \citep{Bang-et-al-2018,Ghosal-deviation-2019}. For example,
consider the problem of selecting the best performing states from
a given set of states having the same maximal fidelity. Now as far
as fidelity is concerned, all states are equally good for quantum
teleportation but note that fidelity deviations of the states would
vary in general. Because we want fidelity deviation to be as small
as possible, the best performing states are clearly those with the
minimum fidelity deviation. Indeed, it has been shown \citep{Ghosal-deviation-2019}
that within a set of states having the same maximal fidelity, larger
than the classical bound, one can always find states with zero fidelity
deviation. Here we take this idea forward and apply it to a more general
setting. 

Specifically, we focus on the problem of characterizing two-qubit
states that are optimal for quantum teleportation for a given value
of a physical (state) property with a well-defined measure, where
the optimal states are defined (see the next subsection) by considering
both maximal fidelity and fidelity deviation. The state properties
considered for our paper are purity, Bell nonlocality, and entanglement,
with the respective measures being linear entropy $L$ \citep{Bose-Vedral-2000,Wei-2003-MEM},
the maximum mean value $B$ of the Bell-CHSH observable \citep{RP_Horodecki-1996},
and concurrence $C$ \citep{Concurrence-1998}. 

It is important to note that, for a given two-qubit state $\rho$,
none of the state properties, $L\left(\rho\right)$, $B\left(\rho\right)$,
and $C\left(\rho\right)$, is a proper quantifier for quantum teleportation,
although each of them can provide meaningful information. For example,
$\rho$ is useful for quantum teleportation (under trace-preserving
LOCC) iff $C\left(\rho\right)>0$ \citep{VV-2003}; if $\rho$ violates
the generalized Bell-CHSH inequality then it is useful for teleportation
\citep{RMP-Horodecki-96} (the converse does not hold); $\rho$ is
not useful if $L\left(\rho\right)\geq L_{0}$ \citep{Bose-Vedral-2000},
where $L_{0}$ is the classical threshold. Thus there appears to be
some connection between the state properties and the ability of a
state to perform a quantum information processing task, such as quantum
teleportation. But to the best of our knowledge, such connections
have not been explored in a comprehensive manner, especially in the
context of optimality of resource states, although very recently \citep{Nandi+-2018},
for two-qubit $X$ states maximal fidelity was studied for a given
purity (mixedness) and concurrence. In the present paper, we will
show that under a very reasonable definition of what constitutes an
optimal state the state parameters (of optimal states) have a definite
functional relationship with the given value of the state property
under consideration. Note that this does not imply that the state
properties actually quantify quantum teleportation, but once we fix
their values they can completely (holds for $L$ and $B$), or almost
completely (holds for $C$), characterize the optimal states. 

\subsection*{Results }

Let us now explain what we mean by optimal states for a given value
of state property, which we assume to have a well-defined measure,
say $P$. For a given value $P=\mathcal{P}$, denote the set of all
two-qubit states by $S\left(\mathcal{P}\right)$. This set is defined
as
\begin{eqnarray}
S\left(\mathcal{P}\right) & = & \left\{ \rho\vert P\left(\rho\right)=\mathcal{P}\right\} .\label{set}
\end{eqnarray}
Now define the largest maximal fidelity 
\begin{eqnarray}
F_{\mathcal{P}} & = & \max_{\rho\in S\left(\mathcal{P}\right)}F_{\rho},\label{LMF}
\end{eqnarray}
where $F_{\rho}$ denotes the maximal fidelity of $\rho$ and the
maximum is taken over all states in $S\left(\mathcal{P}\right)$.
Let us denote the set of states with $F_{\rho}=F_{\mathcal{P}}$ by
$S\left(F_{\mathcal{P}}\vert\mathcal{P}\right)$. Now it might be
the case that for the given value $\mathcal{P}$, $F_{\mathcal{P}}$
does not exceed the classical bound, which is $\frac{2}{3}$ for qubits
\citep{MPR-Horodecki-1999,Badziag-2000,Massar-Popescu-classical-1995}.
Clearly, such states are not useful for quantum teleportation, and
we will therefore disregard any such value or values of $P$ from
our consideration. In fact, we will see that such a situation arises
for linear entropy. 

Without loss of generality let us therefore assume that $F_{\mathcal{P}}>\frac{2}{3}$.
Now from the result proved in \citep{Ghosal-deviation-2019} we know
$S\left(F_{\mathcal{P}}\vert\mathcal{P}\right)$ must always contain
states with zero fidelity deviation, and we call such states optimal.
The optimal states therefore form a subset of $S\left(F_{\mathcal{P}}\vert\mathcal{P}\right)$.
Intuitively, it appears that the optimal set should be a proper subset
of $S\left(F_{\mathcal{P}}\vert\mathcal{P}\right)$, but it turns
out this is not always the case (discussed below). Thus to summarize,
a given state $\rho\in S\left(\mathcal{P}\right)$ is optimal if and
only if $F_{\rho}=F_{\mathcal{P}}$ and $\Delta_{\rho}=0$, where
$\Delta_{\rho}$ denotes the fidelity deviation of $\rho$. Note that
the definition of optimal states is completely general and holds for
any physical property with a well-defined measure. 

The main contribution of this paper lies in complete characterization
of\emph{ }two-qubit states that are optimal for quantum teleportation
for a given $L$, $B$, and $C$, respectively\emph{.} For purity
and Bell nonlocality we identify the largest maximal fidelity states
by solving constrained optimization problems. Interestingly, in both
cases we find that the states with the largest maximal fidelity have
zero fidelity deviation (the converse is not true in general); hence,
they are optimal. In other words, the set of optimal states is identical
to the set of largest maximal fidelity states in these two cases.
For concurrence, however, a similar approach fails to work (explained
later). So here we take a different approach where we make use of
the results in \citep{VV-2-2002-concurrence} to identify the set
of optimal states. In this case we find that, unlike linear entropy
and Bell-CHSH violation, the optimal states for a given concurrence
form a strict subset of the states with the largest maximal fidelity;
that is, within the set of largest maximal fidelity states only some
states have zero fidelity deviation and, therefore, only those states
are optimal. In addition, we provide the necessary and sufficient
conditions a two-qubit state $\rho$ must satisfy so that it is optimal
for quantum teleportation for $\text{\ensuremath{\mathcal{P}}=}P\left(\rho\right)$,
where $P\in\left(L,B,C\right)$. 

For our analysis, we rely on the canonical representation \citep{Badziag-2000,Ghosal-deviation-2019}
of a two-qubit state, which is known to be quite useful in studying
two-qubit nonlocal properties \citep{RMP-Horodecki-96,MPR-Horodecki-1999,Badziag-2000,RP_Horodecki-1996,RM-Horodecki-1996}.
The canonical form is related to the Hilbert-Schmidt representation
\citep{RMP-Horodecki-96,MPR-Horodecki-1999,Badziag-2000,RP_Horodecki-1996,RM-Horodecki-1996}
of a two-qubit state via an appropriate local unitary transformation
\citep{Badziag-2000} and can be described with fewer state parameters.
Note that there is no loss of generality in employing the canonical
form for our studies as the properties that we consider are invariant
under local unitary transformations. 

The rest of the paper is arranged as follows. In Sec.$\,$\ref{prelim}
we review the necessary definitions and concepts we require for this
paper. In particular, we review the Hilbert-Schmidt decomposition
and the canonical form of a two-qubit density matrix, discuss the
notions of maximal fidelity and fidelity deviation, and summarize
the relevant formulas. The main results are derived in Sec.$\,$\ref{main-results},
and we conclude with a short discussion in Sec.$\,$\ref{conclusions}. 

\section{Preliminaries \label{prelim}}

\subsection*{Hilbert-Schmidt representation and the canonical form}

The Hilbert-Schmidt decomposition of a two-qubit density matrix $\rho$
is given by \citep{RMP-Horodecki-96,RP_Horodecki-1996,RM-Horodecki-1996} 

\begin{eqnarray}
\rho & = & \frac{1}{4}\left(I\otimes I+\bm{R}\bm{\cdot\sigma}\otimes I+I\otimes\bm{S}\bm{\cdot\sigma}+\sum_{i,j=1}^{3}T_{ij}\sigma_{i}\otimes\sigma_{j}\right),\label{rho-HSdecomp}
\end{eqnarray}
where $\bm{R}$ and $\bm{S}$ are vectors in $\mathbb{R}^{3}$, $\bm{R}\left(\bm{S}\right)\bm{\cdot}\bm{\sigma}$$=\sum_{i=1}^{3}R_{i}(S_{i})\sigma_{i}$,
and $T_{ij}={\rm Tr}\left(\rho\sigma_{i}\otimes\sigma_{j}\right)$,
where $i,j=1,2,3$, are elements of a real $3\times3$ matrix $T$
(the correlation matrix). 

Let $t_{11},t_{22},t_{33}$ be the eigenvalues of $T$. One can show
\citep{MPR-Horodecki-1999,RP_Horodecki-1996,RM-Horodecki-1996} that
there always exists a product unitary operator $U_{1}\otimes U_{2}$
that transforms $\rho\rightarrow\varrho$, where $\varrho$ is $T$
diagonal. In particular, 
\begin{eqnarray}
\varrho & = & \frac{1}{4}\left(I\otimes I+\bm{r\cdot\sigma}\otimes I+I\otimes\bm{s\cdot\sigma}+\sum_{i=1}^{3}\lambda_{i}\left|t_{ii}\right|\sigma_{i}\otimes\sigma_{i}\right),\label{canonical-rho}
\end{eqnarray}
where $\lambda_{i}\in\left\{ -1,+1\right\} $ and
\begin{eqnarray*}
\bm{r} & = & O_{1}\bm{R},\\
\bm{s} & = & O_{2}\bm{S},\\
T_{\varrho} & = & O_{1}T_{\rho}O_{2}^{\dagger},
\end{eqnarray*}
for unique $3\times3$ rotation matrices $O_{1}$ and $O_{2}$ obtained
via 
\begin{eqnarray}
U_{i}\bm{n\cdot\sigma}U_{i}^{\dagger} & = & \left(O_{i}^{\dagger}\bm{n}\right)\bm{\cdot\sigma}\;\;i=1,2.\label{U-O-relation}
\end{eqnarray}
One can further choose $U_{1}$ and $U_{2}$ such that (a) if $\det T\leq0$
then $\lambda_{i}=-1$ for $\left|t_{ii}\right|\neq0$, $i=1,2,3$;
(b) if $\det T>0$ then $\lambda_{i},\lambda_{j}=-1$, $\lambda_{k}=+1$
for any choice of $i\neq j\neq k\in\left\{ 1,2,3\right\} $ satisfying
$\left|t_{ii}\right|\geq\left|t_{jj}\right|\geq\left|t_{kk}\right|$.
The transformed state $\varrho$ is defined as the canonical form
of $\rho$ \citep{Badziag-2000,Ghosal-deviation-2019} \footnote{This definition, taken from \citep{Ghosal-deviation-2019}, differs
from that of \citep{Badziag-2000}. But the difference is only cosmetic
as both representations are related to each other by local unitary
operations. We prefer our representation, however (see \citep{Ghosal-deviation-2019})!}. 

\subsection*{Teleportation fidelity and fidelity deviation}

The average teleportation fidelity (average fidelity) for a two-qubit
state $\rho$ is defined as \citep{RMP-Horodecki-96} 
\begin{eqnarray}
\left\langle f_{\rho}\right\rangle  & = & \int f_{\psi,\rho}{\rm d}\psi,\label{teleportationl-fidelity}
\end{eqnarray}
where $f_{\psi,\rho}=\left\langle \psi\left|\varsigma\right|\psi\right\rangle $
is the fidelity between an input-output pair $\left(\left|\psi\rangle\langle\psi\right|,\varsigma\right)$
and the integral is over a uniform distribution ${\rm d}\psi$ (normalized
Haar measure, $\int{\rm d}\psi=1$) of input states $\psi=\left|\psi\right\rangle \left\langle \psi\right|$.
Unless stated otherwise, the average fidelity is computed with respect
to the standard protocol \citep{teleportation-93}. Note that, $\frac{2}{3}\leq\left\langle f_{\rho}\right\rangle \leq1$
, where the equality $\left\langle f_{\rho}\right\rangle =1$ holds
if and only if $\rho$ is maximally entangled. 

Fidelity deviation, which is a measure of fidelity fluctuations, is
defined as the standard deviation of fidelity over all input states
\citep{Bang-et-al-2012,Bang-et-al-2018,Ghosal-deviation-2019}: 
\begin{eqnarray}
\delta_{\rho} & = & \sqrt{\left\langle f_{\rho}^{2}\right\rangle -\left\langle f_{\rho}\right\rangle ^{2}},\label{fidelity-standard-deviation}
\end{eqnarray}
where $\left\langle f_{\rho}^{2}\right\rangle =\int f_{\psi,\rho}^{2}{\rm d}\psi$.
Note that, 
\begin{eqnarray*}
\delta_{\rho}^{2} & \leq & \left\langle f_{\rho}\right\rangle -\left\langle f_{\rho}\right\rangle ^{2}=\left\langle f_{\rho}\right\rangle \left(1-\left\langle f_{\rho}\right\rangle \right)\leq\frac{1}{4}.
\end{eqnarray*}
Thus $0\leq\delta_{\rho}\leq\frac{1}{2}$, where $\delta_{\rho}=0$
iff $f_{\psi,\rho}=\left\langle f_{\rho}\right\rangle $ for all $\left|\psi\right\rangle $. 

\subsubsection*{Maximal fidelity and fidelity deviation }

For a given two-qubit state $\rho$ the maximal fidelity $F_{\rho}$
is defined as the maximal value of the average fidelity obtained over
all strategies within the standard protocol and local unitary operations
\citep{RMP-Horodecki-96,Badziag-2000}. It can be shown that \citep{Badziag-2000}
\begin{eqnarray*}
F_{\rho} & = & F_{\varrho},\\
 & = & \left\langle f_{\varrho}\right\rangle ,
\end{eqnarray*}
where $\varrho$ is the canonical representative of $\rho$. The equality
$F_{\rho}=\left\langle f_{\varrho}\right\rangle $ indicates an optimal
strategy that consists of two steps: first, transform $\rho\rightarrow\varrho$
using an appropriate local unitary operation, and then use $\varrho$
for quantum teleportation following the standard protocol \citep{Badziag-2000}.

The fidelity deviation corresponding to the optimal protocol mentioned
above is defined as $\Delta_{\rho}=\delta_{\varrho}$ \citep{Ghosal-deviation-2019}.
In general, the minimum of $\delta_{\rho}$, where the minimum is
taken over all local unitary strategies, is not the same as $\Delta_{\rho}$.
While it is possible to minimize $\delta_{\rho}$ over all local unitary
strategies, such strategies might not always achieve the maximal value
$F_{\rho}$. So, purely for physical reasons, the definition is an
appropriate one (for details, see \citep{Ghosal-deviation-2019}). 

The states with vanishing fidelity deviation are of special interest.
Such states are said to satisfy the universality condition \citep{Bang-et-al-2012,Bang-et-al-2018,Ghosal-deviation-2019}:
the condition that all input states are teleported equally well. For
example, maximally entangled states and Werner states \citep{Bang-et-al-2018}
have this property but there also exist other states that are universal
(for a detailed discussion on the universality condition and other
examples see \citep{Ghosal-deviation-2019}). 
\begin{defn*}
A two-qubit state $\rho$ is useful for quantum teleportation iff
$F_{\rho}>\frac{2}{3}$ \citep{RMP-Horodecki-96,MPR-Horodecki-1999}
and universal iff $\Delta_{\rho}=0$ \citep{Ghosal-deviation-2019}. 
\end{defn*}
The useful and universal conditions can hold independent of each other.
In particular, a state can be useful but not universal, and vice versa.
The necessary and sufficient condition for a two-qubit state to be
both useful and universal is given in \citep{Ghosal-deviation-2019}. 

The results in \citep{RMP-Horodecki-96,Badziag-2000,Ghosal-deviation-2019}
have established that both $F_{\rho}$ and $\Delta_{\rho}$ are functions
of the eigenvalues of the $T$ matrix. The table below summarizes
the formulas and the conditions under which they hold \citep{Badziag-2000,Ghosal-deviation-2019}: 
\begin{center}
{\small{}}%
\begin{tabular}{|c|c|c|}
\hline 
{\small{}$\rho$} & {\small{}$F_{\rho}$} & {\small{}$\Delta_{\rho}$}\tabularnewline
\hline 
\hline 
{\small{}det$T<0$} & {\small{}$\frac{1}{2}\left(1+\frac{1}{3}\stackrel[i=1]{3}{\sum}\left|t_{ii}\right|\right)>\frac{2}{3}\;$iff$\,\stackrel[i=1]{3}{\sum}\left|t_{ii}\right|>1$} & {\small{}$\frac{1}{3\sqrt{10}}\sqrt{\stackrel[i<j=1]{3}{\sum}\left(\left|t_{ii}\right|-\left|t_{jj}\right|\right)^{2}}$}\tabularnewline
\hline 
{\small{}det$T=0$} & {\small{}$\frac{1}{2}\left(1+\frac{1}{3}\stackrel[i=1]{3}{\sum}\left|t_{ii}\right|\right)\leq\frac{2}{3}$} & {\small{}$\frac{1}{3\sqrt{10}}\sqrt{\stackrel[i<j=1]{3}{\sum}\left(\left|t_{ii}\right|-\left|t_{jj}\right|\right)^{2}}$}\tabularnewline
\hline 
{\small{}det$T>0$} & {\small{}$\frac{1}{2}\left[1+\frac{1}{3}\underset{i\neq j\neq k}{\max}\left(\left|t_{ii}\right|+\left|t_{jj}\right|-\left|t_{kk}\right|\right)\right]\leq\frac{2}{3}$} & {\small{}$\underset{i\neq j\neq k}{\min}$$\frac{1}{3\sqrt{10}}\sqrt{\left(\left|t_{ii}\right|-\left|t_{jj}\right|\right)^{2}+\left(\left|t_{ii}\right|+\left|t_{kk}\right|\right)^{2}+\left(\left|t_{jj}\right|+\left|t_{kk}\right|\right)^{2}}$}\tabularnewline
\hline 
\end{tabular}{\small{}}\\
\par\end{center}

Let us briefly summarize the important points: 
\begin{itemize}
\item A two-qubit state is useful for quantum teleportation iff $\sum_{i=1}^{3}\left|t_{ii}\right|>1$
\citep{RMP-Horodecki-96}; if a state is useful then it also has the
property $\det T<0$, but the converse is not true in general, for
example, there exist entangled states with $\det T<0$ but for which
$\sum_{i=1}^{3}\left|t_{ii}\right|\leq1$ \citep{Ghosal-deviation-2019}. 
\item Not all entangled two-qubit states are useful in the sense that there
exist entangled states for which $\sum_{i=1}^{3}\left|t_{ii}\right|\leq1$.
But for all such states one can apply suitable trace-preserving LOCC
to make them useful \citep{VV-2003,Badziag-2000,SB-2002}. 
\end{itemize}

\section{Results \label{main-results}}

Recall that for a given value of a state property the optimal states
are those with the largest maximal fidelity and zero fidelity deviation;
as explained earlier, we shall disregard the values of the state property
for which the largest maximal fidelity does not exceed the classical
bound. 

From the table we see that the states with $\det T\geq0$ are not
useful because $F_{\rho}\leq\frac{2}{3}$. Thus it suffices to focus
only on states with $\det T<0$, and for such states the maximal fidelity
and the fidelity deviation are given by 
\begin{eqnarray}
F_{\rho} & = & \frac{1}{2}\left(1+\frac{1}{3}\sum_{i=1}^{3}\left|t_{ii}\right|\right),\label{F-rho}\\
\Delta_{\rho} & = & \frac{1}{3\sqrt{10}}\sqrt{\sum_{i<j=1}^{3}\left(\left|t_{ii}\right|-\left|t_{jj}\right|\right)^{2}.}\label{Delta-rho}
\end{eqnarray}
It holds that $\Delta_{\rho}=0$ iff $\left|t_{ii}\right|$, $i=1,2,3$
are all equal \citep{Ghosal-deviation-2019}. Note that, $\det T<0$
implies $t_{ii}\neq0$ for all $i=1,2,3$.

As we explained in the introduction, our analysis will be carried
out for a canonical $\varrho$. For states with $\det T<0$ the canonical
$\varrho$ is given by 
\begin{eqnarray}
\varrho & = & \frac{1}{4}\left(I\otimes I+\bm{r\cdot\sigma}\otimes I+I\otimes\bm{s\cdot\sigma}-\sum_{i=1}^{3}\left|t_{ii}\right|\sigma_{i}\otimes\sigma_{i}\right).\label{canonical-rho-detT<0}
\end{eqnarray}

\subsection{Optimal two-qubit states for a given linear entropy }

The purity of a state can be measured by the linear entropy $L$,
which is a linear approximation of the von Neumann entropy. The normalized
linear entropy for a two-qubit state $\rho$ is defined as \citep{Bose-Vedral-2000,Wei-2003-MEM}
\begin{eqnarray}
L\left(\rho\right) & = & \frac{4}{3}\left(1-\text{Tr}\rho^{2}\right),\label{Linear-entropy}
\end{eqnarray}
which is zero for pure states and one for the maximally mixed state.
Since $L\left(\rho\right)=L\left(\varrho\right)$, one obtains 
\begin{eqnarray}
L\left(\varrho\right) & = & 1-\frac{1}{3}\sum_{i=1}^{3}\left(r_{i}^{2}+s_{i}^{2}+\left|t_{ii}\right|^{2}\right),\label{linear-entropy-expression}
\end{eqnarray}
for $\varrho$ given by (\ref{canonical-rho-detT<0}). 

The goal is to find the optimal states for a given linear entropy
$L=\mathcal{L}$. The first step therefore is to identify the states
with the largest maximal fidelity. This requires us to maximize $F_{\varrho}$
as given by (\ref{F-rho}) for fixed $\mathcal{L}$. The constrained
optimization problem can be stated as
\begin{eqnarray}
{\rm maximize} & \;\;\; & \sum_{i=1}^{3}\left|t_{ii}\right|\label{N=00005Crho}\\
{\rm such\,}{\rm that} & \;\;\; & 1-\frac{1}{3}\sum_{i=1}^{3}\left(r_{i}^{2}+s_{i}^{2}+\left|t_{ii}\right|^{2}\right)=\mathcal{L}.\label{linearentropyconstraint}
\end{eqnarray}
One immediately notices that the states that maximize $\sum_{i=1}^{3}\left|t_{ii}\right|$
must have $r_{i}=s_{i}=0$ for all $i=1,2,3$ -- this condition,
however, is only necessary and not sufficient. Note, however, that
the conditions $r_{i}=s_{i}=0$ for $i=1,2,3$ imply that the local
vectors $\bm{R},\bm{S}$ associated with $\rho$ must also be zero. 

Now setting $r_{i}=s_{i}=0$, $i=1,2,3,$ (\ref{N=00005Crho}) and
(\ref{linearentropyconstraint}) become 
\begin{eqnarray}
{\rm maximize} & \;\;\; & \sum_{i=1}^{3}\left|t_{ii}\right|\label{N=00005Crho-1}\\
{\rm such\,}{\rm that} & \;\;\; & \sum_{i=1}^{3}\left|t_{ii}\right|^{2}=3\left(1-\mathcal{L}\right).\label{linearentropyconstraint-1}
\end{eqnarray}
This can be easily solved. First, we parametrize $\left|t_{ii}\right|$
for $i=1,2,3$ as
\begin{eqnarray*}
\left|t_{11}\right| & = & A\sin\theta\cos\phi,\\
\left|t_{22}\right| & = & A\sin\theta\sin\phi,\\
\left|t_{33}\right| & = & A\cos\theta,
\end{eqnarray*}
where $A=\sqrt{3\left(1-\mathcal{L}\right)}$ is constant and $\theta\in\left(0,\frac{\pi}{2}\right)$,
$\phi\in\left(0,\frac{\pi}{2}\right)$ or $\theta\in\left(\frac{3\pi}{2},2\pi\right)$,
$\phi\in\left(\pi,\frac{3\pi}{2}\right)$. With this parametrization
the problem reduces to finding the maxima of the function 
\begin{eqnarray*}
f\left(\theta,\phi\right) & = & A\left(\sin\theta\cos\phi+\sin\theta\sin\phi+\cos\theta\right)
\end{eqnarray*}
within the acceptable ranges of $\theta$ and $\phi$ as specified
before. A simple calculation shows the maxima are obtained at two
critical points: $\theta^{*}=\tan^{-1}\sqrt{2}$, $\phi^{*}=\frac{\pi}{4}$,
and $\theta^{*}=2\pi-\tan^{-1}\sqrt{2}$, $\phi^{*}=\frac{5\pi}{4}$,
and at both of them 
\begin{eqnarray}
\left|t_{ii}\right| & = & \sqrt{1-\mathcal{L}},\;i=1,2,3.\label{opt-t_ii-LE}
\end{eqnarray}
Thus the largest maximal fidelity $F_{\mathcal{L}}$ for given linear
entropy $\mathcal{L}$ is 
\begin{eqnarray}
F_{\mathcal{L}} & = & \frac{1}{2}\left(1+\sqrt{1-\mathcal{L}}\right).\label{F-opt-LE-1}
\end{eqnarray}
The states that achieve this maximal value have $\bm{R}=\bm{S}=0$
and $\left|t_{ii}\right|=\sqrt{1-\mathcal{L}}$, $i=1,2,3$. Since
$\left|t_{ii}\right|$ are all equal, the largest maximal fidelity
states have zero fidelity deviation. 

Now, what remains to be checked is whether for all values of $\mathcal{L}$,
$0\leq\mathcal{L}\leq1$ the largest maximal fidelity states are useful,
i.e., $F_{\mathcal{L}}>\frac{2}{3}$. It turns out that they are not
\citep{Bose-Vedral-2000}. The condition $F_{\mathcal{L}}>\frac{2}{3}$
holds only when $0\leq\mathcal{L}<\frac{8}{9}$. In other words, not
all values of the linear entropy are permissible. 

We can now summarize the results. 
\begin{prop}
For a given linear entropy $\mathcal{L}$, where $0\leq\mathcal{L}<\frac{8}{9}$,
the optimal two-qubit states for quantum teleportation are those with
local vectors $\bm{R}=\bm{S}=0$ and $\left|t_{ii}\right|=\sqrt{1-\mathcal{L}}$,
$i=1,2,3$. 
\end{prop}
The Hilbert-Schmidt decomposition of an optimal state $\rho_{{\rm opt}}$
is given by 
\begin{eqnarray}
\rho_{{\rm opt}} & = & \frac{1}{4}\left(I\otimes I+\sum_{i,j=1}^{3}T_{ij}\sigma_{i}\otimes\sigma_{j}\right),\label{rho-HSdecomp-1}
\end{eqnarray}
where the eigenvalues $t_{ii}$ of $T$ are such that $\left|t_{ii}\right|=\sqrt{1-\mathcal{L}}$,
$i=1,2,3$. This is explicitly reflected in the canonical form 
\begin{eqnarray}
\varrho_{{\rm opt}} & = & \frac{1}{4}\left(I\otimes I-\sqrt{1-\mathcal{L}}\sum_{i=1}^{3}\sigma_{i}\otimes\sigma_{i}\right).\label{opt-rho-LE}
\end{eqnarray}
So far, our analysis has been completely general. Now we consider
a state-specific question: Given a two-qubit state $\rho$ with linear
entropy $L\left(\rho\right)$, is $\rho$ optimal for $L\left(\rho\right)=\mathcal{L}$?
The following proposition answers this question. 
\begin{prop}
Let $\rho$ be a two-qubit state of linear entropy $\mathcal{L}$,
where $0\leq\mathcal{L}<\frac{8}{9}$. Then $F_{\rho}=F_{\mathcal{L}}$
and $\Delta_{\rho}=0$ if and only if $\left|t_{ii}\right|=\sqrt{1-\mathcal{L}}$
for $i=1,2,3$. 
\end{prop}
\begin{proof}
If $\left|t_{ii}\right|=\sqrt{1-\mathcal{L}}$, $i=1,2,3$, then from
(\ref{F-rho}) we get $F_{\rho}=\frac{1}{2}\left(1+\sqrt{1-\mathcal{L}}\right)$,
which is indeed the largest maximal fidelity $F_{\mathcal{L}}$ for
a given $\mathcal{L}$, and, moreover, $\left|t_{ii}\right|$ are
all equal, hence, $\Delta_{\rho}=0$. On the other hand, if $F_{\rho}=F_{\mathcal{L}}$
then from Eqns.$\,$(\ref{F-rho}) and (\ref{F-opt-LE-1}) we find
that 
\begin{eqnarray*}
\sum_{i=1}^{3}\left|t_{ii}\right| & = & 3\sqrt{1-\mathcal{L}}.
\end{eqnarray*}
Now we impose the condition $\Delta_{\rho}=0$. This implies $\left|t_{ii}\right|$,
$i=1,2,3$ are all equal. Then from the above equation we get $\left|t_{ii}\right|=\sqrt{1-\mathcal{L}}$,
$i=1,2,3$. This completes the proof. 
\end{proof}

\subsection{Optimal Bell-nonlocal states }

The violation of the Bell-CHSH inequality indicates the presence of
Bell-nonlocality. Now whether a two-qubit state $\rho$ violates the
Bell-CHSH inequality or not is completely determined by the function
$M\left(\rho\right)$ = $\max_{i>j}\left(t_{ii}^{2}+t_{jj}^{2}\right)$
\citep{RP_Horodecki-1996}. 
\begin{prop}
\citep{RP_Horodecki-1996} A two-qubit state $\rho$ violates the
Bell-CHSH inequality if and only if $M\left(\rho\right)>1$. 
\end{prop}
The function $M\left(\rho\right)$ is related to the maximal mean
value $B$ of the Bell-CHSH observable via the relation $B=2\sqrt{M\left(\rho\right)}$
\citep{RP_Horodecki-1996}; hence, $M\left(\rho\right)>1$ implies
$B>2$ -- the condition for Bell-CHSH violation. 

For a canonical $\varrho$ with $\det T<0$, the eigenvalues of the
$T$ matrix are given by $-\left|t_{ii}\right|$, $i=1,2,3$. Therefore,
\begin{eqnarray*}
M\left(\varrho\right) & = & \max_{i>j}\left(\left|t_{ii}\right|^{2}+\left|t_{jj}\right|^{2}\right),\\
 & = & M\left(\rho\right).
\end{eqnarray*}
We want to identify the optimal Bell-nonlocal states for a given value
of $B$, say $\mathcal{B}$, where $\mathcal{B}>2$. Thus for a given
$\mathcal{B}$, the corresponding $M$ is fixed; let us denote this
fixed value by $\mathcal{M}$, where $\mathcal{M}>1$. Since the relation
$\left|t_{ii}\right|\geq\left|t_{jj}\right|\geq\left|t_{kk}\right|$
always holds for some choice of $i,j,k\in\left\{ 1,2,3\right\} $,
$i\neq j\neq k$, the constrained optimization problem can be stated
as 
\begin{eqnarray}
{\rm maximize} & \;\;\; & \sum_{i=1}^{3}\left|t_{ii}\right|\label{N=00005Crho-2}\\
{\rm such}\;{\rm that} & \;\;\; & \left|t_{ii}\right|^{2}+\left|t_{jj}\right|^{2}=\mathcal{M},\label{Bell-constraint}
\end{eqnarray}
where $\left\{ \left|t_{ii}\right|,\left|t_{jj}\right|\right\} \geq\left|t_{kk}\right|$
for $i\neq j\neq k\in\left\{ 1,2,3\right\} $ \footnote{For example: if $\left|t_{11}\right|>\left|t_{33}\right|>\left|t_{22}\right|$
then (\ref{Bell-constraint}) becomes $\left|t_{11}\right|^{2}+\left|t_{33}\right|^{2}=\mathcal{M}$
etc.}. Note that $\mathcal{M}$ does not depend on $t_{kk}$. Let us now
parametrize $\left|t_{ii}\right|$ and $\left|t_{jj}\right|$ by 
\begin{eqnarray*}
\left|t_{ii}\right| & = & \sqrt{\mathcal{M}}\cos\theta,\\
\left|t_{jj}\right| & = & \sqrt{\mathcal{M}}\sin\theta,
\end{eqnarray*}
where $\theta\in\left(0,\frac{\pi}{2}\right)$. Thus in order to solve
our problem, first we need to maximize the function
\begin{eqnarray*}
f\left(\theta\right) & = & \sqrt{\mathcal{M}}\left(\cos\theta+\sin\theta\right)
\end{eqnarray*}
for $\theta\in\left(0,\frac{\pi}{2}\right)$. The above function has
one critical point: $\theta^{*}=\frac{\pi}{4}$ in the range $\left(0,\frac{\pi}{2}\right)$,
and at this point 
\begin{equation}
f\left(\theta^{*}\right)=\sqrt{2\mathcal{M}},
\end{equation}
which is in fact the maximum. Consequently, $\left|t_{ii}\right|=\left|t_{jj}\right|=\sqrt{\frac{\mathcal{M}}{2}}$. 

So now we are left with the problem of maximizing $\sqrt{2\mathcal{M}}+\left|t_{kk}\right|$.
To maximize $\sqrt{2\mathcal{M}}+\left|t_{kk}\right|$ we can take
$\left|t_{kk}\right|$ as large as possible provided $|t_{kk}|\leq|t_{ii}|=|t_{jj}|$.
But we have already shown $\left|t_{ii}\right|=\left|t_{jj}\right|=\sqrt{\frac{\mathcal{M}}{2}}$.
Therefore, the maximum is obtained for $\left|t_{kk}\right|=\sqrt{\frac{\mathcal{M}}{2}}$. 

Thus the states that maximize the maximal fidelity for given $\mathcal{B}$
are those with $\left|t_{ii}\right|=\frac{\mathcal{B}}{2\sqrt{2}}$,
$i=1,2,3$. The largest maximal fidelity is given by 
\begin{eqnarray}
F_{\mathcal{B}} & = & \frac{1}{2}\left(1+\frac{\mathcal{B}}{2\sqrt{2}}\right)\label{maxF-Bell-violation}
\end{eqnarray}
Note here that $F_{\mathcal{B}}>\frac{2}{3}$ for $\mathcal{B}>2$.
So the largest maximal fidelity states are always useful for quantum
teleportation for all $\mathcal{B}>2$. Further, $\left|t_{ii}\right|$,
$i=1,2,3$ are all equal; hence, the largest maximal fidelity states
have zero fidelity deviation. 
\begin{prop}
The optimal two-qubit states for quantum teleportation for a given
$\mathcal{B}>2$ are those with $\left|t_{ii}\right|=\frac{\mathcal{B}}{2\sqrt{2}}$,
$i=1,2,3$. 
\end{prop}
The Hilbert-Schmidt decomposition of an optimal two-qubit state $\rho_{{\rm opt}}$
is given by 
\begin{eqnarray}
\rho_{{\rm opt}} & = & \frac{1}{4}\left(I\otimes I+\bm{R}\bm{\cdot\sigma}\otimes I+I\otimes\bm{S}\bm{\cdot\sigma}+\sum_{i,j=1}^{3}T_{ij}\sigma_{i}\otimes\sigma_{j}\right),\label{rho-HSdecomp-2}
\end{eqnarray}
where the eigenvalues $t_{ii}$ of $T$ are such that $\left|t_{ii}\right|=\frac{\mathcal{B}}{2\sqrt{2}}$
for $i=1,2,3$, and the canonical form can be expressed as 
\begin{eqnarray}
\varrho_{{\rm opt}} & = & \frac{1}{4}\left(I\otimes I+\bm{r\cdot\sigma}\otimes I+I\otimes\bm{s\cdot\sigma}-\frac{\mathcal{B}}{2\sqrt{2}}\sum_{i=1}^{3}\sigma_{i}\otimes\sigma_{i}\right).\label{canonical-opt-Bell}
\end{eqnarray}
Note that, unlike the optimal states for a given linear entropy, where
the local vectors need to be zero, no particular condition is being
imposed on the local vectors in this case. 

Now we would like to know whether a given state $\rho$ is optimal
for quantum teleportation for $B\left(\rho\right)=\mathcal{B}>2$.
We have the following proposition. 
\begin{prop}
Let $\rho$ be a two-qubit state for which the maximum mean value
of the Bell-CHSH observable is $\mathcal{B}>2$. Then $F_{\rho}=F_{\mathcal{B}}$
and $\Delta_{\rho}=0$ if and only if $\left|t_{ii}\right|=\frac{\mathcal{B}}{2\sqrt{2}}$
for $i=1,2,3$. 
\end{prop}

\subsection{Optimal states for a given entanglement}

Now we consider the problem of characterizing the optimal two-qubit
states for a given concurrence \citep{Concurrence-1998}. The concurrence
for a two-qubit state $\rho$ is defined as \citep{Concurrence-1998}
\begin{eqnarray}
C\left(\rho\right) & = & \max\left(0,a_{1}-a_{2}-a_{3}-a_{4}\right),\label{concurrence}
\end{eqnarray}
where $a_{1}\geq a_{2}\geq a_{3}\geq a_{4}$ are the square roots
of the eigenvalues of the matrix $\tilde{\rho}$ defined as
\begin{eqnarray*}
\tilde{\rho} & = & \rho\left(\sigma_{2}\otimes\sigma_{2}\right)\rho^{*}\left(\sigma_{2}\otimes\sigma_{2}\right),
\end{eqnarray*}
where $\sigma_{2}$ is the Pauli bit-phase flip matrix, and $\rho^{*}$
denotes the complex conjugation of $\rho$ in the computational basis.
Note that, $0\leq C\left(\rho\right)\leq1$. 
\begin{fact}
A two-qubit state $\rho$ is entangled iff $C\left(\rho\right)>0$. 
\end{fact}
For a two-qubit state $\rho$, $C\left(\rho\right)$ can be exactly
computed if we have complete knowledge of the state. However, $C\left(\rho\right)$
cannot be expressed in terms of the state parameters for a generic
$\rho$, or even for a canonical $\varrho$, except for some special
classes of states. Thus the previous method of obtaining the largest
maximal fidelity states is no longer useful. So we take a different
approach. 

From the results in \citep{VV-2-2002-concurrence} one can easily
show that the maximal fidelity $F_{\rho}$ of a two-qubit entangled
state $\rho$ is bounded above by 
\begin{equation}
F_{\rho}\leq\frac{2+N\left(\rho\right)}{3}\leq\frac{2+C\left(\rho\right)}{3},\label{upper bounds on fidelity}
\end{equation}
where $N\left(\rho\right)$ is the entanglement negativity, defined
as 
\begin{eqnarray}
N\left(\rho\right) & = & -2\lambda_{\min}\left(\rho^{\Gamma}\right),\label{negativity}
\end{eqnarray}
where $\lambda_{\min}$ is the smallest eigenvalue of the partial
transposed matrix $\rho^{\Gamma}$. The inequalities (\ref{upper bounds on fidelity})
are saturated for certain classes of states, and for such states $N\left(\rho\right)=C\left(\rho\right)$. 
\begin{prop}
\label{LMF-C} The largest maximal fidelity $F_{\mathcal{C}}$ achievable
for a given concurrence $\mathcal{C}$ is 
\begin{eqnarray}
F_{\mathcal{C}} & = & \frac{2+\mathcal{C}}{3}.\label{F(C)=00003Dgiven-concurrence}
\end{eqnarray}
\end{prop}
\begin{rem*}
The above equation shows that $F_{\mathcal{C}}>\frac{2}{3}$, whenever
$\mathcal{C}>0$. Thus the largest maximal fidelity states for a given
concurrence are always useful for quantum teleportation \footnote{This is particularly noteworthy because not all entangled states are
directly useful for quantum teleportation, but the states with the
largest maximal fidelity for any given concurrence $C>0$ always are. }. 
\end{rem*}
The following lemma, proved in \citep{VV-2-2002-concurrence}, provides
the necessary and sufficient condition for Eq.$\,$(\ref{F(C)=00003Dgiven-concurrence})
to hold. 
\begin{lem}
\label{saturation} Let $\rho$ be a two-qubit state of concurrence
$C\left(\rho\right)>0$. Then $F_{\rho}=F_{C\left(\rho\right)}$,
where $F_{C\left(\rho\right)}$ is the largest maximal fidelity achievable
by all two-qubit states of concurrence equal to $C\left(\rho\right)$,
if and only if 
\begin{eqnarray}
\rho^{\Gamma}\left|\Psi\right\rangle  & = & \lambda_{\min}\left|\Psi\right\rangle ,\label{N=000026C-condition}
\end{eqnarray}
where $\left|\Psi\right\rangle $ is a maximally entangled state. 
\end{lem}
Note that $\lambda_{\min}=-\frac{C\left(\rho\right)}{2}$ when the
above eigenvalue equation holds. This is because negativity equals
concurrence when the inequalities are saturated. 

The implication of the above lemma can be understood as follows: For
a given two-qubit entangled state $\rho$, let $\lambda_{\min}$ be
the minimum eigenvalue (negative, of course) of $\rho^{\Gamma}$ and
let $\left|\phi\right\rangle $ be the corresponding eigenvector,
which may or may not be maximally entangled. Then Lemma \ref{saturation}
tells us that if $\left|\phi\right\rangle $ is maximally entangled,
then $N\left(\rho\right)=C\left(\rho\right)$ and $F_{\rho}=F_{C\left(\rho\right)}$,
but if, on the other hand, $\left|\phi\right\rangle $ is not maximally
entangled, then $N\left(\rho\right)<C\left(\rho\right)$ and $F_{\rho}<F_{C\left(\rho\right)}$.
In particular, when $\left|\phi\right\rangle $ is maximally entangled
we get two important pieces of information without explicit calculations:
first, we get to know the concurrence of the state and second, we
get to know the maximal fidelity of the state. 

So whether the equality $F_{\rho}=F_{C\left(\rho\right)}$ holds or
not is completely determined by the nature of the eigenvector $\left|\phi\right\rangle $.
But there is no easy way to check this without explicit calculations
which may be complicated for arbitrary states. Nevertheless, we will
show that there is an efficient way, one that invokes the canonical
representation of $\rho$. 

First we have an important lemma. 
\begin{lem}
\label{t_ii=00003D2C+1} Let $\rho$ be a two-qubit state $\rho$
of concurrence $C\left(\rho\right)>0$. Then $F_{\rho}=F_{C\left(\rho\right)}$
if and only if 
\begin{eqnarray}
\sum_{i=1}^{3}\left|t_{ii}\right| & = & 2C\left(\rho\right)+1.\label{t(ii)related-to-Concurrence}
\end{eqnarray}
\end{lem}
The proof follows from (\ref{F-rho}) and (\ref{F(C)=00003Dgiven-concurrence}).
The lemma tells us that, for a given two-qubit state $\rho$ of concurrence
$C\left(\rho\right)>0$, if (\ref{t(ii)related-to-Concurrence}) holds
then $F_{\rho}=F_{C\left(\rho\right)}$, and if it does not then $F_{\rho}<F_{C\left(\rho\right)}$. 

Let us now look at a useful consequence of the above lemma.
\begin{cor}
\label{corollary =00005Clambda_min} Let $\rho$ be a two-qubit state
of concurrence $C\left(\rho\right)>0$ satisfying the eigenvalue equation
(\ref{N=000026C-condition}). Then 
\begin{eqnarray}
\lambda_{\min} & = & -\frac{1}{4}\left(\sum_{i=1}^{3}\left|t_{ii}\right|-1\right).\label{lambda-min}
\end{eqnarray}
\end{cor}
\begin{proof}
Suppose a two-qubit state $\rho$ of concurrence $C\left(\rho\right)>0$
satisfies the eigenvalue equation (\ref{N=000026C-condition}). Then
$F_{\rho}=F_{C\left(\rho\right)}$. From Lemma \ref{t_ii=00003D2C+1}
we know that for such a state $\sum_{i=1}^{3}\left|t_{ii}\right|=2C\left(\rho\right)+1$
holds, and therefore, $C\left(\rho\right)=\frac{\stackrel[i=1]{3}{\sum}\left|t_{ii}\right|-1}{2}$.
Noting that $\lambda_{\min}=-\frac{C\left(\rho\right)}{2}$ we get
the desired expression.
\end{proof}
We now show an efficient way to find whether the eigenvalue equation
(\ref{N=000026C-condition}) is satisfied. First, we prove that a
similar eigenvalue equation holds for any other density matrix related
to $\rho$ by some local unitary transformation, where the eigenvalue
remains the same as before but the corresponding eigenvector, although
maximally entangled, is different in general. 
\begin{lem}
\label{unitary-relation} Let $\rho$ be a two-qubit state of concurrence
$C\left(\rho\right)>0$ with the property $F_{\rho}=F_{C\left(\rho\right)}$.
Let $\rho^{\prime}$ be another two-qubit state given by $\rho^{\prime}=\left(U\otimes V\right)\rho\left(U^{\dagger}\otimes V^{\dagger}\right)$,
where $U$, $V$ are unitary operators. Then 
\begin{eqnarray}
\left(\rho^{\prime}\right)^{\Gamma}\left|\Psi^{\prime}\right\rangle  & = & \lambda_{\min}\left|\Psi^{\prime}\right\rangle ,\label{unitary-equivalent-eigenvalue-equation}
\end{eqnarray}
where $\left|\Psi^{\prime}\right\rangle =\left(U\otimes V^{T}\right)\left|\Psi\right\rangle $
is a maximally entangled state. 
\end{lem}
The proof is given in the appendix. 

Thus the unitary freedom in the eigenvalue equation (\ref{N=000026C-condition}),
as reflected in (\ref{unitary-equivalent-eigenvalue-equation}), leaves
open the possible existence of some useful $\rho^{\prime}$ for which
finding the eigenvector becomes an easy task. Indeed, for a canonical
$\varrho$ the eigenvector is uniquely determined in the sense that
it no longer depends on $\rho$. 
\begin{lem}
\label{varrho-N+Scondition} Let $\varrho$ be the canonical form
of a two-qubit state $\rho$ of concurrence $C\left(\rho\right)>0$.
Then $F_{\rho}=F_{C\left(\rho\right)}$ if and only if 
\begin{eqnarray}
\varrho^{\Gamma}\left|\Phi^{+}\right\rangle  & = & \lambda_{\min}\left|\Phi^{+}\right\rangle ,\label{Prop-5}
\end{eqnarray}
where $\left|\Phi^{+}\right\rangle =\frac{1}{\sqrt{2}}\left(\left|00\right\rangle +\left|11\right\rangle \right)$
and $\lambda_{\min}=-\frac{1}{4}\left(\stackrel[i=1]{3}{\sum}\left|t_{ii}\right|-1\right)$. 
\end{lem}
\begin{proof}
Let $F_{\rho}=F_{C\left(\rho\right)}$. Then the eigenvalue equation
(\ref{N=000026C-condition}) is satisfied, and, by virtue of Corollary
\ref{corollary =00005Clambda_min}, $\lambda_{\min}=-\frac{1}{4}\left(\sum_{i=1}^{3}\left|t_{ii}\right|-1\right)$.
Since $\varrho$ is related to $\rho$ through local unitary operation,
Eq.$\,$(\ref{unitary-equivalent-eigenvalue-equation}) is also satisfied
for the same $\lambda_{\min}$ but for a different maximally entangled
eigenvector, say $\left|\Psi^{\prime\prime}\right\rangle $. Now given
that $C\left(\rho\right)>0$ and $F_{\rho}=F_{C\left(\rho\right)}$,
$\rho$ must have the property $\det T<0$. Thus the canonical $\varrho$
is described by (\ref{canonical-rho-detT<0}). In the appendix we
show that, for such a canonical $\varrho$, if $\varrho^{\Gamma}\left|\phi\right\rangle =-\frac{1}{4}\left(\sum_{i=1}^{3}\left|t_{ii}\right|-1\right)\left|\phi\right\rangle $,
where $\left|\phi\right\rangle $ is a normalized pure state, then
$\left|\phi\right\rangle =\left|\Phi^{+}\right\rangle $. This proves
the first part of the lemma. 

Now suppose the eigenvalue equation (\ref{Prop-5}) is satisfied,
where $\varrho$ is the canonical form of a two-qubit state of $\rho$
of concurrence $C\left(\rho\right)>0$. Since $\rho$ is entangled,
so is $\varrho$, and therefore, $\lambda_{\min}<0$. Then according
to Lemma \ref{unitary-relation} $\lambda_{\min}$ is also the minimum
eigenvalue of $\rho^{\Gamma}$ with the corresponding eigenvector
being maximally entangled, and therefore, the condition in Lemma \ref{saturation}
is satisfied. Hence, $F_{\rho}=F_{C\left(\rho\right)}$ and $\lambda_{\min}=-\frac{C\left(\rho\right)}{2}$.
From Corollary \ref{corollary =00005Clambda_min} we now conclude
that $\lambda_{\min}=-\frac{1}{4}\left(\sum_{i=1}^{3}\left|t_{ii}\right|-1\right)$.
This completes the proof. 
\end{proof}
Let us now find the conditions under which the eigenvalue equation
(\ref{Prop-5}) holds. 
\begin{lem}
\label{r+s=00003D0} The eigenvalue equation (\ref{Prop-5}) holds
if and only if $\bm{r}+\bm{s}=0$, where $\bm{r},\bm{s}$ are the
local vectors of $\varrho$. 
\end{lem}
The proof is given in the appendix. 

Lemma \ref{r+s=00003D0} only gives us a necessary condition the states
with the largest maximal fidelity must satisfy for any given concurrence.
But Lemma \ref{r+s=00003D0} together with Lemma \ref{t_ii=00003D2C+1}
provide the necessary and sufficient conditions for the largest maximal
fidelity states for any given concurrence $\mathcal{C}>0$. 
\begin{thm}
Let $\rho$ be a two-qubit state. Then, for a given concurrence $\mathcal{C}>0$,
$C\left(\rho\right)=\mathcal{C}$ and $F_{\rho}=F_{\mathcal{C}}$
provided the conditions
\begin{eqnarray*}
\bm{r}+\bm{s} & = & 0,\\
\sum_{i=1}^{3}\left|t_{ii}\right| & = & 2\mathcal{C}+1
\end{eqnarray*}
are met simultaneously, where $\bm{r},\bm{s}$ are the local vectors
associated with $\varrho$ -- the canonical form of $\rho$. 
\end{thm}
\begin{proof}
Let $\rho$ be a two-qubit state with $t_{ii}$, $i=1,2,3$ being
the eigenvalues of the $T$ matrix and let $\varrho$ be the canonical
form of $\rho$ with $\bm{r}$ and $\bm{s}$ being the local vectors.
Assume that both equations in the theorem are satisfied. 

Given that $\mathcal{C}>0$, the second equation implies $\sum_{i=1}^{3}\left|t_{ii}\right|>1$.
Thus $\rho$ is entangled (and so is, $\varrho$) and has the property
$\det T<0$. Therefore $\varrho$ admits the form given by (\ref{canonical-rho-detT<0}).
Now $\bm{r}+\bm{s}=0$ implies that (from Lemma \ref{r+s=00003D0})
for such a canonical $\varrho$ the eigenvalue equation 
\begin{eqnarray*}
\varrho^{\Gamma}\left|\Phi^{+}\right\rangle  & = & -\frac{1}{4}\left(\sum_{i=1}^{3}\left|t_{ii}\right|-1\right)\left|\Phi^{+}\right\rangle 
\end{eqnarray*}
is satisfied, where the eigenvalue must be negative because $\stackrel[i=1]{3}{\sum}\left|t_{ii}\right|>1$.
Now the partial transposed matrix of a two-qubit entangled state has
exactly one negative eigenvalue. Therefore, $-\frac{1}{4}\left(\sum_{i=1}^{3}\left|t_{ii}\right|-1\right)$
is the minimum eigenvalue of $\varrho^{\Gamma}$. Then from Lemma
\ref{varrho-N+Scondition} we conclude that $F_{\rho}=F_{C\left(\rho\right)}$.
But we know that for any state $\rho$ of concurrence $C\left(\rho\right)$
if the equality $F_{\rho}=F_{C\left(\rho\right)}$ holds, then $\sum_{i=1}^{3}\left|t_{ii}\right|=2C\left(\rho\right)+1$.
Hence, $C\left(\rho\right)=\mathcal{C}$. 
\end{proof}
It is clear that the fidelity deviation of a state with the largest
maximal fidelity is nonzero in general because to satisfy Eq.$\,\left(\ref{t(ii)related-to-Concurrence}\right)$
the absolute values of $t_{ii}$, $i=1,2,3$ need not be equal. 

By definition, the optimal states for a given concurrence are the
states with the largest maximal fidelity and zero fidelity deviation.
Now that we have already identified the largest maximal fidelity states
for any given concurrence $\mathcal{C}>0$, we can apply the condition
for zero fidelity deviation which demands that $\left|t_{ii}\right|$,
$i=1,2,3$ are all equal. One can now use Eq.$\,$(\ref{t(ii)related-to-Concurrence})
to obtain $\left|t_{ii}\right|=\frac{2\mathcal{C}+1}{3}$, $i=1,2,3$.
Thus the canonical form of an optimal state for a fixed concurrence
$\mathcal{C}>0$ is given by 
\begin{eqnarray}
\varrho_{{\rm opt}} & = & \frac{1}{4}\left(I\otimes I+\bm{r}\cdot\bm{\sigma}\otimes I-I\otimes\bm{r}\cdot\bm{\sigma}-\frac{2\mathcal{C}+1}{3}\sum_{i=1}^{3}\sigma_{i}\otimes\sigma_{i}\right),\label{HS-X-1-1}
\end{eqnarray}
where we have used the fact that $\bm{r}+\bm{s}=0$ and $\left|t_{ii}\right|=\frac{2\mathcal{C}+1}{3}$
for all $i=1,2,3$. 

\section{Conclusions \label{conclusions}}

The average fidelity \citep{RMP-Horodecki-96,MPR-Horodecki-1999}
is generally considered to be the measure of choice to characterize
quantum teleportation. But recently \citep{Bang-et-al-2018,Ghosal-deviation-2019}
it has been emphasized that in addition to the average fidelity one
should also take into account fidelity deviation, which is defined
as the standard deviation of fidelity over all input states and serves
as a well-defined measure of fluctuations in fidelity. For two-qubit
states the maximal average fidelity (maximal fidelity) \citep{RMP-Horodecki-96,Badziag-2000}
and the corresponding fidelity deviation \citep{Ghosal-deviation-2019}
are known, where both are given by simple formulas that can also be
exactly computed. So for two-qubit states a comprehensive characterization
of quantum teleportation is possible. 

In \citep{Ghosal-deviation-2019} it was pointed out that fidelity
deviation can serve as a useful filter to select the optimal states
for quantum teleportation from a known set of states. In particular,
for a given set of states, where every state in the set has the same
maximal fidelity, the most desirable states are those with zero fidelity
deviation. In this paper, we applied this idea in a more general setting.
Specifically, we characterized two-qubit optimal states -- the states
with the largest maximal fidelity and zero fidelity deviation --
for a given linear entropy $L$, the maximum mean value $B$ of the
Bell-CHSH observable, and concurrence $C$, respectively. For our
analysis, we extensively used the canonical description of a two-qubit
density matrix. This greatly simplified calculations, especially in
the cases of purity and Bell-nonlocality for which we were able to
find the largest maximal fidelity states by solving appropriate constrained
optimization problems. On the other hand, for a given entanglement,
we had to consider a different approach altogether for reasons explained
earlier. 

We found that for given purity and Bell-CHSH violation, respectively,
the largest maximal fidelity states also exhibit zero fidelity deviation,
and therefore, they are optimal. For a given concurrence, however,
not all largest maximally fidelity states have zero fidelity deviation;
thus the optimal states form a strict subset of the largest maximal
fidelity states. The optimal states in general have the following
properties: 
\begin{itemize}
\item The eigenvalues of the correlation matrix $T$ are functions of the
given value of the state property under consideration. 
\item The local vectors satisfy certain conditions (except, for optimal
Bell-nonlocal states): in case of linear entropy the local vectors
must be zero, and for entanglement the sum of the local vectors appearing
in the canonical form must vanish. 
\end{itemize}
Are there new insights to be gained from the aforementioned results?
We believe so and there are two in our opinion. First of all, the
results show that state properties can have an essential role in characterizing
optimal states, which, in fact, holds for the ones we considered.
Second, our analysis widens the perspective about quantum teleportation
and right away shows some counter-intuitive consequences. For example,
if one considers a pure state, a Werner derivative state, and a Werner
state of the same concurrence, they all yield the same maximal fidelity.
Therefore, in terms of fidelity alone, these states are all equally
good. Our characterization involving fidelity deviation, however,
lifts this degeneracy and Werner states emerge as the clear winner,
followed by Werner derivative states, and, quiet surprisingly, the
pure state is the worst among the three. 

Let us now discuss possible applications of our results. The state
properties that we considered, namely, entanglement, Bell nonlocality,
and purity, are often used to characterize resources in quantum information
processing tasks. In practical situations, however, one should try
to find the 'best performing quantum states' within available resources.
Our paper addresses this question within the set of two-qubit states,
and moreover, our characterization involving fluctuations moves us
closer to addressing issues during practical (experimental) implementations
of quantum teleportation -- for example, teleportation might be implemented
as an intermediate step in a quantum circuit to be later processed
by some gates that are typically sensitive to fluctuations of their
inputs (see, for example, \citep{Pedersen-2008}) and here, if one
would have used just the average fidelity the analysis would have
simply been less complete. So we believe that the results will be
beneficial for future experimental implementations of quantum teleportation. 

In summary, the primary objective in any quantum information protocol
is to use available resources in best possible ways. The classification
of resource states for quantum teleportation via average fidelity
provides the first step of segregating good resource states, and as
our results demonstrate, fidelity deviation provides another layer
of filtration to sieve out better states from a set of states that
are equally good in terms of fidelity. The results presented here
identify the best possible resource states from a set of states with
fixed concurrence, or linear entropy, or Bell violation, and we hope
that results of this kind would help us to better understand quantum
teleportation in relation to state properties of resource states. 

Finally, the results in our previous work \citep{Ghosal-deviation-2019}
and the present paper show that fidelity deviation in many ways complements
maximal fidelity, which has long been regarded as the sole figure
of merit for quantum teleportation. So far analyses \citep{Bang-et-al-2018,Ghosal-deviation-2019}
have been confined to two-qubit states only. It would be interesting
to study how these results could be generalized in higher dimensions
and also in many-qubit systems, especially in network structures using
quantum repeaters. 
\begin{acknowledgement*}
DD acknowledges financial support from University Grants Commission
(UGC), Government of India. SB is supported in part by SERB (Science
and Engineering Research Board), Department of Science and Technology,
Government of India through Project No. EMR/2015/002373. 
\end{acknowledgement*}

\section*{Appendix}

\subsection{Proof of Lemma \ref{unitary-relation}}

From the given relation 
\begin{eqnarray*}
\rho^{\prime} & = & \left(U\otimes V\right)\rho\left(U^{\dagger}\otimes V^{\dagger}\right)
\end{eqnarray*}
one can express $\rho$ as 
\begin{eqnarray*}
\rho & = & \left(U^{\dagger}\otimes V^{\dagger}\right)\rho^{\prime}\left(U\otimes V\right).
\end{eqnarray*}
Then the eigenvalue equation
\begin{eqnarray*}
\rho^{\Gamma}\left|\Psi\right\rangle  & = & \lambda_{\min}\left|\Psi\right\rangle ,
\end{eqnarray*}
 can be written as
\begin{eqnarray*}
\left[\left(U^{\dagger}\otimes V^{\dagger}\right)\rho^{\prime}\left(U\otimes V\right)\right]^{\Gamma}\left|\Psi\right\rangle  & = & \lambda_{\min}\left|\Psi\right\rangle .
\end{eqnarray*}
Taking the partial transpose with respect to the second subsystem,
the above equation can be written in the form
\begin{eqnarray*}
\left[\left(U^{\dagger}\otimes V^{*}\right)\left(\rho^{\prime}\right)^{\Gamma}\left(U\otimes V^{T}\right)\right]\left|\Psi\right\rangle  & = & \lambda_{\min}\left|\Psi\right\rangle .
\end{eqnarray*}
Now multiplying both sides with $\left(U\otimes V^{T}\right)$ from
the left we obtain 
\begin{eqnarray*}
\left(\rho^{\prime}\right)^{\Gamma}\left(U\otimes V^{T}\right)\left|\Psi\right\rangle  & = & \lambda_{\min}\left(U\otimes V^{T}\right)\left|\Psi\right\rangle .
\end{eqnarray*}
Denoting $\left|\Psi^{\prime}\right\rangle =\left(U\otimes V^{T}\right)\left|\Psi\right\rangle $
we arrive at (\ref{unitary-equivalent-eigenvalue-equation}). This
completes the proof. 

\subsection{Proof of part of Lemma \ref{varrho-N+Scondition} }

Let $\varrho$ be the canonical form of a two-qubit density matrix
$\rho$ of concurrence $C\left(\rho\right)$. Since $F_{\rho}=F_{\mathcal{C}}$,
we know that Lemma \ref{saturation} is satisfied. Then according
to Lemma \ref{unitary-relation} the following eigenvalue equation
holds: 
\begin{eqnarray}
\varrho^{\Gamma}\left|\Psi^{\prime\prime}\right\rangle  & = & -\frac{1}{4}\left(\sum_{i=1}^{3}\left|t_{ii}\right|-1\right)\left|\Psi^{\prime\prime}\right\rangle ,\label{Proposition-4-equation-1}
\end{eqnarray}
where $\left|\Psi^{\prime\prime}\right\rangle $ is a maximally entangled
state and $t_{ii}$, $i=1,2,3$, are the eigenvalues of $T$ associated
with $\rho$. We will show that $\left|\Psi^{\prime\prime}\right\rangle =\left|\Phi^{+}\right\rangle $,
where $\left|\Phi^{+}\right\rangle =\frac{1}{\sqrt{2}}\left(\left|00\right\rangle +\left|11\right\rangle \right)$. 

Now suppose the eigenvalue equation 
\begin{eqnarray}
\varrho^{\Gamma}\left|\phi\right\rangle  & = & -\frac{1}{4}\left(\sum_{i=1}^{3}\left|t_{ii}\right|-1\right)\left|\phi\right\rangle \label{eigen-equation}
\end{eqnarray}
is satisfied for some normalized two-qubit pure state $\left|\phi\right\rangle $,
where $\varrho$ is given by (\ref{canonical-rho-detT<0}). Then the
identity 
\begin{eqnarray}
\left\langle \phi\left|\varrho^{\Gamma}\right|\phi\right\rangle  & = & \frac{1}{4}\left(1-\sum_{i=1}^{3}\left|t_{ii}\right|\right)\label{expectation=00003Deigenvalue}
\end{eqnarray}
holds. Note that while (\ref{eigen-equation}) implies (\ref{expectation=00003Deigenvalue})
the converse in general does not hold because, while the expectation
value will always produce a number, it may not be the eigenvalue. 

Now any two-qubit pure state $\left|\phi\right\rangle $ can be written
as a linear combination of the four Bell states 
\begin{eqnarray*}
\left|\phi\right\rangle  & = & a_{1}\left|\Phi^{+}\right\rangle +a_{2}\left|\Phi^{-}\right\rangle +a_{3}\left|\Psi^{+}\right\rangle +a_{4}\left|\Psi^{-}\right\rangle ,
\end{eqnarray*}
where $a_{i}\in\mathbb{C},$$i=1,\dots,4$, $\sum_{i=1}^{4}\left|a_{i}\right|^{2}=1$,
and the Bell-states are given by 
\begin{eqnarray*}
\left|\Phi^{\pm}\right\rangle  & = & \frac{1}{\sqrt{2}}\left(\left|00\right\rangle \pm\left|11\right\rangle \right),\\
\left|\Psi^{\pm}\right\rangle  & = & \frac{1}{\sqrt{2}}\left(\left|01\right\rangle \pm\left|10\right\rangle \right).
\end{eqnarray*}
The LHS of (\ref{expectation=00003Deigenvalue}) leads to 
\begin{eqnarray*}
\left\langle \phi\left|\varrho^{\Gamma}\right|\phi\right\rangle  & = & \frac{1}{4}+\frac{1}{4}\left[\left\langle \phi\left|\bm{r}\cdot\bm{\sigma}\otimes I\right|\phi\right\rangle +\left\langle \phi\left|I\otimes\left(\bm{s}\cdot\bm{\sigma}\right)^{T}\right|\phi\right\rangle \right]\\
 &  & +\frac{1}{4}\left[\left|t_{11}\right|\left(-\left|a_{1}\right|^{2}+\left|a_{2}\right|^{2}-\left|a_{3}\right|^{2}+\left|a_{4}\right|^{2}\right)\right]\\
 &  & +\frac{1}{4}\left[\left|t_{22}\right|\left(-\left|a_{1}\right|^{2}+\left|a_{2}\right|^{2}+\left|a_{3}\right|^{2}-\left|a_{4}\right|^{2}\right)\right]\\
 &  & +\frac{1}{4}\left[\left|t_{33}\right|\left(-\left|a_{1}\right|^{2}-\left|a_{2}\right|^{2}+\left|a_{3}\right|^{2}+\left|a_{4}\right|^{2}\right)\right],
\end{eqnarray*}
where we have taken the partial transposition with respect to the
second qubit. Now the second term on the right hand side is a function
of $\bm{r},\text{\ensuremath{\bm{s}}}$, so it must be equal to zero
because the RHS of (\ref{expectation=00003Deigenvalue}) does not
contain any term which is a function of the local vectors. So we end
up with the task of solving the following four equations 
\begin{eqnarray*}
-\left|a_{1}\right|^{2}+\left|a_{2}\right|^{2}-\left|a_{3}\right|^{2}+\left|a_{4}\right|^{2} & = & -1\\
-\left|a_{1}\right|^{2}+\left|a_{2}\right|^{2}+\left|a_{3}\right|^{2}-\left|a_{4}\right|^{2} & = & -1\\
-\left|a_{1}\right|^{2}-\left|a_{2}\right|^{2}+\left|a_{3}\right|^{2}+\left|a_{4}\right|^{2} & = & -1\\
\left|a_{1}\right|^{2}+\left|a_{2}\right|^{2}+\left|a_{3}\right|^{2}+\left|a_{4}\right|^{2} & = & +1,
\end{eqnarray*}
where the last equation is due to the normalization condition. The
above equations can be conveniently expressed in the matrix form as
\begin{eqnarray*}
\left(\begin{array}{cccc}
-1 & 1 & -1 & 1\\
-1 & 1 & 1 & -1\\
-1 & -1 & 1 & 1\\
1 & 1 & 1 & 1
\end{array}\right)\left(\begin{array}{c}
\left|a_{1}\right|^{2}\\
\left|a_{2}\right|^{2}\\
\left|a_{3}\right|^{2}\\
\left|a_{4}\right|^{2}
\end{array}\right) & = & \left(\begin{array}{c}
-1\\
-1\\
-1\\
1
\end{array}\right).
\end{eqnarray*}
The $4\times4$ matrix on the left is invertible; hence, 
\begin{eqnarray*}
\left(\begin{array}{c}
\left|a_{1}\right|^{2}\\
\left|a_{2}\right|^{2}\\
\left|a_{3}\right|^{2}\\
\left|a_{4}\right|^{2}
\end{array}\right) & = & \left(\begin{array}{cccc}
-1 & 1 & -1 & 1\\
-1 & 1 & 1 & -1\\
-1 & -1 & 1 & 1\\
1 & 1 & 1 & 1
\end{array}\right)^{-1}\left(\begin{array}{c}
-1\\
-1\\
-1\\
1
\end{array}\right)\\
 & = & \frac{1}{4}\left(\begin{array}{cccc}
-1 & -1 & -1 & 1\\
1 & 1 & -1 & 1\\
-1 & 1 & 1 & 1\\
1 & -1 & 1 & 1
\end{array}\right)\left(\begin{array}{c}
-1\\
-1\\
-1\\
1
\end{array}\right)\\
 & = & \left(\begin{array}{c}
1\\
0\\
0\\
0
\end{array}\right)
\end{eqnarray*}
Thus we have shown that $\left|\psi\right\rangle =\left|\Phi^{+}\right\rangle $
(up to some irrelevant global phase). 

\subsection{Proof of Lemma \ref{r+s=00003D0}}

Let us first obtain an expression for $\varrho^{\Gamma}\left|\Phi^{+}\right\rangle $,
where $\varrho$ is given by (\ref{canonical-rho-detT<0}), and the
partial transposition is taken with respect to the second qubit. 
\begin{eqnarray*}
\varrho^{\Gamma}\left|\Phi^{+}\right\rangle  & = & \frac{1}{4}\left(I\otimes I+\bm{r\cdot\sigma}\otimes I+I\otimes\left(\bm{s\cdot\sigma}\right)^{T}-\sum_{i=1}^{3}\left|t_{ii}\right|\sigma_{i}\otimes\sigma_{i}^{T}\right)\left|\Phi^{+}\right\rangle \\
 & = & \frac{1}{4}\left[\left\{ \bm{r\cdot\sigma}\otimes I+I\otimes\left(\bm{s\cdot\sigma}\right)^{T}\right\} \left|\Phi^{+}\right\rangle +\left(1-\sum_{i=1}^{3}\left|t_{ii}\right|\right)\left|\Phi^{+}\right\rangle \right]\\
 & = & \frac{1}{4}\left[\left(r_{1}+s_{1}\right)\left|\Psi^{+}\right\rangle -i\left(r_{2}+s_{2}\right)\left|\Psi^{-}\right\rangle +\left(r_{3}+s_{3}\right)\left|\Phi^{-}\right\rangle +\left(1-\sum_{i=1}^{3}\left|t_{ii}\right|\right)\left|\Phi^{+}\right\rangle \right].
\end{eqnarray*}
Therefore, if $\left|\Phi^{+}\right\rangle $ is an eigenvector of
$\varrho^{\Gamma}$ it holds that $r_{i}+s_{i}=0$ for all $i=1,2,3$.
Conversely, if $r_{i}+s_{i}=0$ for all $i=1,2,3$, it immediately
follows that $\left|\Phi^{+}\right\rangle $ is the eigenvector of
$\varrho^{\Gamma}$ with eigenvalue $-\frac{1}{4}\left(\stackrel[i=1]{3}{\sum}\left|t_{ii}\right|-1\right)$.
This completes the proof.

\end{document}